\documentclass[12pt]{article}
\usepackage{amsfonts}
\usepackage{amsmath}
\usepackage{mathrsfs}
\usepackage{amssymb}
\usepackage{color,amscd,array,graphicx,mathdots,bm,amsfonts,epsfig,amsmath}
\usepackage{amsthm}
\usepackage{enumerate}
\usepackage{geometry}

\usepackage{caption}
\usepackage{subfigure}

\newtheorem{theorem}{Theorem}[section]

\newtheorem{definition}{Definition}[section]
\newtheorem{lemma}{Lemma}[section]
\newtheorem{proposition}{Proposition}[section]
\newtheorem{remark}{Remark}[section]

\usepackage{graphicx,xypic,color,float} 

\usepackage[colorlinks=true,citecolor=blue]{hyperref}

\makeatletter
\newcommand{\figcaption}{\def\@captype{figure}\caption}
\newcommand{\tabcaption}{\def\@captype{table}\caption}
\makeatother

\def\de{{\delta}}
\def\ga{{\gamma}}

\newfam\msbmfam\font\tenmsbm=msbm10\textfont
\msbmfam=\tenmsbm\font\sevenmsbm=msbm7
\scriptfont\msbmfam=\sevenmsbm

\def\EE{\mathbb E}\def\HH{\mathbb H}\def\PP{\mathbb P}
\def\RR{\mathbb R}
\def\PP{\mathbb P}

\def\cG{{\cal G}}
\def\cM{{\cal M}}

\def\<{\left<}\def\>{\right>}

\def\({\left(}\def\){\right)}

\usepackage{graphicx} 
\usepackage{epstopdf}
\usepackage{float}

\begin{document}
\title{Mean-variance portfolio selection under partial information with drift uncertainty}
\author{Jie Xiong\footnote{Department of Mathematics and SUSTech International center for Mathematics, Southern University of Science and Technology, Shenzhen, 518055, China. 
This author is supported by National Natural Science Foundation of China grants 61873325 and 11831010 and Southern University of Science and Technology Start up fund Y01286220. Email: xiongj@sustech.edu.cn.
} \and 
Zuo Quan Xu\thanks{Department of Applied Mathematics, The Hong Kong Polytechnic University, Kowloon, Hong Kong. The author acknowledges financial supports from NSFC (No.11971409), Hong Kong GRF (No.15204216 and No.15202817), and the Hong Kong Polytechnic University. Email: maxu@polyu.edu.hk.} 
\and Jiayu Zheng\footnote{Corresponding author. Department of Math and Stat Science, University of Alberta, Edmonton, Canada. This author is supported by National Natural Science Foundation of China grant 11901598. Email: jy\_zheng@outlook.com.} }
\date{}
\maketitle
\begin{abstract}
In this paper, we  study the mean-variance portfolio selection problem under partial information with drift uncertainty. First we show that the market model is   complete even in this case while the information is not complete and the drift is uncertain. Then, the optimal  strategy based on partial information is derived, which reduces to solving a related backward stochastic differential equation (BSDE). Finally, we propose an efficient numerical scheme to approximate the optimal portfolio that is the solution of the BSDE mentioned above. Malliavin calculus and the particle representation play important roles in this scheme.
\bigskip
\end{abstract}



\section{Introduction}
The mean-variance portfolio selection model pioneered by Markowitz \cite{MM} has paved the foundation for modern portfolio theory and has been widely applied in financial economics. Markowitz proposed and solved the problem in a single period setting. For half of a century, however, the optimal dynamic mean-variance portfolio selection problem was not solved due to the non-separable structure of the variance minimization problem in the sense of dynamic programming. This difficulty was finally overcame by Li and Ng \cite{LN} and Zhou and Li \cite{ZL} via an embedding scheme, for multi-period and continuous-time cases, respectively. Since then, many scholars have devoted their attentions to the study of the dynamic extensions of the Markowitz model, see, for example, Li et al. \cite{LZL}, Lim and Zhou \cite{LimZ}, Zhou and Yin \cite{ZY}, Hu and Zhou \cite{HZ}, Bielecki et al. \cite{BJPZ}, Li and Zhou \cite{LZ}, Chiu and Li \cite{CL} in continuous-time settings. All these works assume that the Brownian motions that are driving the stock prices are completely observable to the investors. In reality, however, the driving Brownian motions are often not observable to the investors, and the stock prices are the only observable information based on which the investors make the decisions. This fact motivates the study of the so-called partial information portfolio selection problem. Xiong and Zhou \cite{XZ} established the separation principle to separate the filtering and optimization problems for the mean-variance portfolio selection problem with partial information. They also developed analytical and numerical approaches in obtaining the filter as well as solving the related backward stochastic differential equation. 

The optimal redeeming problem of stock loans under drift uncertainty has been studied by Xu and Yi \cite{XY}. In their model, the inherent uncertainty of the trend of the stock is modeled by a two-state random variable representing bull and bear trends, respectively; the current trend of the stock is not known to the investor so that she/he has to make decisions based on partial information. They derive the optimal redeeming strategies based on the prediction of the stock trend.

In this paper, we study a mean-variance problem under partial information with drift uncertainty. Our contributions to the literature are summarized below: First, we show that the market model is   complete even in this case while the information is not complete and the drift is uncertain. Second, 
the optimal strategy based on partial information is derived, which involves the optimal filter of the drift.  Third, an efficient numerical approximation is suggested to solve the BSDE which arises from the mean-variance problem under drift uncertainty. This scheme is  investigated in the context of the Malliavin calculus approach for the
approximation of conditional expectations. We also prove the convergence of our numerical scheme, and study the $L^2$ error induced by the Euler approximation and by the strong law of large number (SLLN).

The rest of the paper is organized as follows. Some preliminary results on filtering and Malliavin calculus are given in Section 2. In Section 3, we derive the innovation process associated with the posteriori probability process of the drift uncertainty model and study its optimization problem under partial information. We also prove the completeness of the market under $\mathcal{G}$ filtration (partial information).  A new numerical scheme is proposed and the asymptotic behavior is studied in Section 4, a couple of numerical results are also presented.

\section{Preliminaries}
\setcounter{equation}{0}

In this section, we state some elementary facts about stochastic filtering and Malliavin calculus for the convenience of the reader. We refer the reader to Sections 8.1-8.3 of Kallianpur \cite{GK} for more details about the general filtering problem and the stochastic equation of the optimal filter, and the book of Nualart and Nualart \cite{Nua} about the Malliavin calculus. 
\par
Let $T$ be a fixed positive constant representing the investment horizon. 
Let $(\Omega, \mathscr{F}, P)$ be a complete probability space and let $\mathcal{F}_t$, $t \in [0, T]$, be an increasing family of sub $\sigma$-fields of $\mathscr{F}$. 
The signal $h_t(\omega)$ and the observation $Z_t(\omega)$, $t \in [0, T]$, are assumed to be two $N$-dimensional processes defined on $(\Omega, \mathscr{F}, P)$ and further related as follows:
\begin{eqnarray}\label{cond1}
Z_t(\omega) = \int_0^t h_u(\omega)du + W_t(\omega),
\end{eqnarray}
where $W_t$ is an $N$-dimensional Wiener process, and $h_t(\omega)$ is a $\RR^{N}$-valued, $(t,\omega)$-measurable function satisfying 
\begin{eqnarray}\label{cond4}
\int_0^T\EE(|h_t|^2)dt <\infty,
\end{eqnarray}
where $|\cdot|$ denotes the Euclidean norm of $N$-dimensional vector. Further, for each $s \in [0, T]$, the $\sigma$-fields 
$\mathscr{F}_s^{h,W}:= \sigma\{h_u, W_u, 0 \le u \le s\}$ and $\mathscr{F}_s^W:= \sigma\{W_{u'} - W_u, s \le u\le u' \le T\}$
are independent. 
Let $\{\mathscr{F}^Z_t\}_{0\leq t\le T}$ be the filtration generated by $Z$. This filtration is called the observation $\sigma$-field. Let $v_t := (v_t^1, \cdots, v_t^N)'$, $t \in [0, T]$, be an $N$-dimensional $\mathscr{F}^Z_t$-adapted innovation process, which is also an $N$-dimensional $\mathscr{F}^Z_t$-adapted Brownian motion.

We list three theorems for ready references. The following one appears in Section 8.3 of \cite{GK} (page 208). 

\begin{theorem}\label{thm2.1}
Under conditions \eqref{cond1} and \eqref{cond4}, every separable, square-integrable $\mathscr{F}^Z_t$-martingale $Y_t$ is sample path continuous a.s. and has the representation 
\begin{eqnarray}
Y_t -\EE(Y_0) = \sum_{i=1}^{N} \int_{0}^{t} \Phi_s^i dv_s^i,\quad t \in [0, T],
\end{eqnarray}
where 
\begin{eqnarray}
\int_0^T \EE(|\Phi_s|^2)ds<\infty
\end{eqnarray}
and $\Phi_s := (\Phi_s^1, \cdots , \Phi_s^N)'$ is jointly measurable and adapted to $\mathscr{F}^Z_t$.
\end{theorem}

The next theorem is called the Clark-Ocone formula (see Theorem 6.1.1 of \cite{Nua}). It expresses a square integrable random variable in terms of the conditional expectation of its Malliavin derivative. Let $B = \(B_t\)_{t \ge 0}$ be a multi-dimensional Brownian motion on a probability space $(\Omega, \mathscr{F}, (\mathscr{F}_t)_{t\geq 0}, P)$, where $(\mathscr{F}_t)_{t\geq 0}$ is the natural filtration of $B$ and $\mathscr{F}=\vee_{t\geq 0}\mathscr{F}_t$. Denote by $D$ the Malliavin derivative operator. We define the Sobolev space $\mathbb{D}^{1,2}$ of random variables as follows:
\[
\mathbb{D}^{1,2}=\bigg\{F\in L^0(\Omega, \mathscr{F}, P): \|F\|_{1,2}^2=\mathbb{E}(|F|^2)+ \mathbb{E}\Big[\int_0^\infty |D_tF|^2dt\Big]<\infty\bigg\},
\]
where $L^0(\Omega, \mathscr{F}, P)$ denotes the set of $\mathscr{F}$-measurable random variables.
\begin{theorem}[Clark-Ocone formula]
Let $F \in \mathbb{D}^{1,2} \cap L^0(\Omega, \mathscr{F}_T, P)$. Then, F admits the following representation
$$ F = \EE (F) + \int_0^T \EE (D_t F | \mathscr{F}_t) dB_t .$$
\end{theorem}

Let $\cM(d,q,\RR)$ be a vector space of matrices with $d$ rows and $q$ columns of $\RR$-valued entries, $\|\cdot\|$ be the canonical Euclidean norm.

Denote by $L^p(0, T; \RR^d)$ the set of all $\RR^d$-valued $\{\mathcal{F}_t\}_{t\in [0,T]}$-adapted processes $f$ in the probability space $(\Omega, \mathcal{F}, \PP)$ whose $L^p$ norm are finite, namely \[\|f\|_{L^p(0, T; \RR^d)}:=\(\EE\int_0^T|f(t)|^pdt \)^{\frac{1}{p}}<\infty.\]
Let $L^p( \mathcal{F},\RR^d)$ be the set of all $\RR^d$-valued random variables $\xi $ with finite $L^p$ norm \[\|\xi\|_{p} := \(\EE |\xi|^p\)^{\frac{1}{p}}<\infty.\]

The next theorem which appears in Section 7 of \cite{PG} (Theorem 7.2), states the error approximation of the Euler scheme for the solution $(X_t)_{t\in[0,T]}$ to a $d$-dimensional stochastic differential equation 
\begin{align}\label{EulerSDE}
dX_t = b(t,X_t)dt + \sigma(t,X_t)dW_t,
\end{align}
where $b: [0,T]\times \RR^d \to \RR^d, \sigma:[0,T] \times \RR^d \to \cM(d,q,\RR)$ are continuous functions, $W = (W_t)_{t\in[0,T]}$ denotes a $q$-dimensional standard Brownian motion defined on a probability space $(\Omega, \mathcal{F}, \PP)$ and $X_0: (\Omega, \mathcal{F}, \PP) \to \RR^d$ is a random vector, independent of $W$.

The discrete time Euler scheme with step $\frac{T}{n}$ is defined by
\begin{align*}
\bar{X}_{t_{k+1}^n} = \bar{X}_{t_{k}^n} + \frac{T}{n}b\(t_k^n, \bar{X}_{t_{k}^n}\) + \sigma \(t_k^n, \bar{X}_{t_{k}^n}\) \sqrt{\frac{T}{n}} Z^n_{k+1}, \ \ &\bar{X}_0 = X_0,\\
&k =1,\cdots,n-1,
\end{align*}
where $t_k^n = \frac{kT}{n}, k =1,\cdots,n$ and $\(Z_k^n\)_{1\le k \le n}$ denotes a sequence of  independent and identically distributed random vectors given by 
\[
Z_k^n := \sqrt{\frac{n}{T}} \( W_{t_k^n} -  W_{t_{k-1}^n}\), \ \ k=1, \cdots, n. 
\]

\begin{theorem}[Strong rate for the Euler scheme]\label{Euler}
Suppose the coefficients $b$ and $\sigma$ of the SDE \eqref{EulerSDE} satisfy the following regularity condition: there exist a real constant $C_{b,\sigma, T}>0$ and an exponent $\beta \in (0,1]$ such that for all $s,t \in [0,T], x, y \in \RR^d$, 
\begin{align}\label{Hbeta}
|b(t,x) - b(s,x)| + \|\sigma(t,x) - \sigma(s,x)\| &\le C_{b,\sigma,T}(1+|x|)|t-s|^{\beta}, \\ 
|b(t,x) - b(t,y)| + \|\sigma(t,x) - \sigma(t,y)\| &\le C_{b,\sigma,T}|y-x|.\label{Hbeta2}
\end{align}
Then for all $p>0$, there exists a universal constant $\kappa_p >0$, depending on $p$ only, such that for every $n\ge T,$
\begin{align}\label{rate}
\big\| \sup_{0 \le k \le n} | X_{t^n_k}- \bar{X}_{t^n_k}^n| \big\|_p \le K(p,b,\sigma, T) \(1+ \|X_0\|_p\) \(\frac{T}{n}\)^{\beta \wedge \frac{1}{2}},
\end{align}
where 
\[
K(p,b,\sigma, T) = \kappa_p C'_{b,\sigma,T} e^{\kappa_p(1+ C'_{b,\sigma,T})^2T}
\]
and
\begin{align}\label{finite}
C'_{b,\sigma,T} = C_{b,\sigma,T} + \max_{t\in[0,T]}|b(t,0)| +\max_{t\in[0,T]}\|\sigma(t,0)\|<+\infty. 
\end{align}
\end{theorem}

\section{Problem formulation}
\setcounter{equation}{0}

\subsection{Model setup}
\numberwithin{equation}{section}

Assume that $(\Omega, \mathcal{F}, P,\{\mathcal{F}_t\}_{t\ge 0})$ is a complete filtered probability space, which represents the financial market. The filtration $\{\mathcal{F}_t\}_{t\ge 0}$ satisfies the usual conditions, and $P$ denotes the probability measure. In this probability space, there exists a standard one-dimensional Brownian motion $W$. The price process of the underlying stock is denoted by $S_t$, $ t \in[0,T]$, which satisfies the stochastic differential equation (SDE):
\begin{eqnarray}\label{drift}
dS_t = \mu S_tdt + S_t dW_t,
\end{eqnarray}
where $\mu$ is random and independent of the Brownian motion $W$, and it may only takes two possible values $a$ and $b$ that satisfy $$\ga := a-b >0.$$
The stock is said to be in its bull trend when $\mu = a$, and in its bear trend when $\mu = b$.

The information up to time $t$ is given by
$$\mathcal{G}_t := \sigma\( S_s: \; s \le t\),\quad t \in[0,T].$$
The \emph{posteriori probability process} $\pi = (\pi_t)_{t \in[0,T]}$ is defined as 
\begin{eqnarray}
\pi_t := P(\mu = a | \cG_t).
\end{eqnarray}
It is used to estimate the chance that the stock is in its bull trend at time $t$. Assume that $0 < \pi_0 <1$. This means it is not clear whether the stock is in its bull trend or bear trend at time $0$. 

Let $u_t$, called a portfolio, be the amount invested in the stock at time $t$.
\begin{definition}
A portfolio (or trading strategy) is called {\bf self-financing} if all the changes of the values of the
portfolio are due to gains or losses realized on investment, that is, no funds are borrowed
or withdrawn from the portfolio at any time. A portfolio $u_{t}$ is called {\bf admissible} if it is $\cG_t$-adapted, self-financing and 
\[\int^T_0\EE(u_t^2)dt<\infty.\]
\end{definition}

Denote by $Y_t$ the wealth process of an agent, and $u_t$ an admissible trading strategy. 
Starting with an initial wealth $y_0>0$, $Y_t$ satisfies the following \emph{wealth equation}:
\begin{align}\label{wealthy01}
\begin{cases}
dY_t = \(\mu u_t + (Y_t - u_t)r \)dt + u_t dW_t, \quad t \in [0,T], \\
Y_0 = y_0.
\end{cases}
\end{align}
where $r$ denotes the interest rate.

Our goal is to solve the following optimization 

{\em Problem (MV)}: To find the best admissible portfolio $u_t$ to 
minimize $\text{Var}(Y_{T})$ subject to the constraint $\EE (Y_{T})= z$, where $Y_t$ is given by 
(\ref{wealthy01}), and $z$ is a given positive number.

Taken as observation, the log-price process $L = (\log S_t)_{t \in[0,T]}$, by It\^{o}'s lemma, satisfies the following SDE
\begin{eqnarray}\label{log}
dL_t = \(\mu - \frac{1}{2} \)dt + dW_t.
\end{eqnarray}
Then, the innovation process
\begin{eqnarray}\label{barW}
\nu_t = L_t - \int_0^t \(b - \frac{1}{2} + \ga \pi_s \) ds
\end{eqnarray}
is a Brownian motion with respect to the observation filtration $\mathcal{G}_t$. It is easy (see \cite{GK}, Chapter 8.1) to verify that $\pi_t$ satisfies the following SDE: 
\begin{align}\label{pi112}
d \pi_t =\ga \pi_t (1 - \pi_t) d\nu_t,\qquad \pi_0=P(\mu=a).
\end{align}
By \eqref{wealthy01} and \eqref{barW}, we get the $\nu_t$-driven representation for $Y$:
\begin{equation}\label{innavation}
\begin{cases}
dY_t = \big(\(b + \ga \pi_t - r\)u_t + r Y_t \big)dt + u_t d\nu_t,\quad t \in[0,T],\\
Y_0 = y_0.
\end{cases}
\end{equation}

\subsection{Optimization under partial information}
\numberwithin{equation}{section}

The optimization problem (MV) turns to minimize $\text{Var}(Y_{T})$ with state equations \eqref{innavation} and the constraint $\EE (Y_{T})= z$.

Let
\begin{eqnarray}\label{rho}
\rho_t := \exp\(- \int_0^t (b-r+\ga \pi_s) d\nu_s -\int_0^t (r+ \frac{1}{2}(b-r+\ga \pi_s)^2)ds \).
\end{eqnarray}
Applying It\^o's formula to $\rho_t$, we get 
\begin{align}\label{rho11}
d\rho_t = -r \rho_t dt - \(b -r +\ga\pi_t\)\rho_t d\nu_t.
\end{align}
Further, applying It\^o's formula to $Y_t\rho_t$, we have 
\[
d(Y_t\rho_t) = \(Y_t\(r\rho_t - \mu \rho_t \) + u_t \rho_t\) d\nu_t.
\]
Therefore, $Y_t\rho_t$ is a $\mathcal{G}_t$-martingale and hence,
\[
\EE (Y_t\rho_t) = y_0.
\]
Denote $Y_T$ by $v$. To find the optimal portfolio, we seek the best $\mathcal{G}_T$-measurable terminal wealth $v$ to minimize the variance 
\begin{eqnarray}\label{opt}
\EE(v -z)^2
\end{eqnarray}
subject to constraints
\begin{eqnarray}
\EE v = z \quad \text{and} \quad \EE(\rho_Tv) = y_0.
\end{eqnarray}

Let $\HH := L^2(\Omega, \mathcal{G}_T, P)$. For $X\in\HH$, let
\[\|X\|_{\HH}:=\(\EE (X^2) \)^{\frac{1}{2}}.\]
Then, $\HH$ is a Hilbert space endowed with the norm $\|\cdot\|_\HH$. Note that 
\[\EE(v -z)^2=\|v-0\|^2_\HH-z^2. \]
Therefore, the optimal $v$ is the projection of $0$ onto the hyperplane 
\[\left\{v\in\HH:\;\EE v=z,\;\EE(v\rho_T)=y_0\right\}.\]

\subsection{Completeness of the market}
\numberwithin{equation}{section}
Denote by $L^2_{\mathcal{G}}(0, T; \RR)$ the set of all $\RR$-valued, $\mathcal{G}_t$-adapted processes $f(t)$ on $[0,T]$ such that \[\EE\int_0^T|f(t)|^2dt <\infty.\] Then $L^2_{\mathcal{G}}(0,T;\RR)$ becomes a Hilbert space endowed with the norm \[\|f\|_{L^2_{\mathcal{G}}(0,T;\RR)}:=\(\EE\int_0^T|f(t)|^2dt \)^{\frac{1}{2}}.\]

\begin{definition}\label{defC}
A contingent claim $v \in \HH$ is called attainable if there is $\Phi_s \in L^2_{\mathcal{G}}(0,T;\mathbb{R})$ such that 
\begin{eqnarray}\label{claim}
v\rho_T = \EE(v\rho_T) + \int_0^T \Phi_s d\nu_s. 
\end{eqnarray}
\end{definition}

Denote the collection of all attainable contingent claims by AC($\mathcal{G}$). Then AC($\mathcal{G}$) is a subspace of $\HH$. Denote by $\HH_0$ the closure of AC($\mathcal{G})$ in $\HH$ under the norm $\|\cdot\|_{\HH}$. 

\begin{definition}
The market is complete if $\HH_0=\mathbb{H}$. 
\end{definition}

The following statement is quite surprise to us. Namely, the market is complete although the information is not complete and the drift is uncertain. It is worth mentioning that completeness was left open in \cite{XZ} for their model. Because of this lacking of completeness result, the authors of \cite{XZ} turn to search the optimal solution $v$ in the space $\HH_0$.

\begin{theorem}\label{Hspace}
The market is complete.
\end{theorem}
\begin{proof}
Since $\HH_0\subseteq \HH$, it suffices to show $\HH\subseteq \HH_{0}$.
For any $V \in \HH$, let $V_n = V\min\{|V|^{-\frac{1}{n}},1\}$. Then 
\[(V_{n}-V)^{2}= V^{2}\mathbf{1}_{|V|>1}(|V|^{-\frac{1}{n}}-1)^{2}\leq V^{2}\mathbf{1}_{|V|>1} \leq V^{2}.\]
Since $V \in \HH$, we have $\EE|V|^2 <\infty$. By the dominated convergent theorem, 
\[\lim_{n\to\infty}\|V_{n}-V\|^{2}_{\HH}=\lim_{n\to\infty}\EE[(V_{n}-V)^{2}]=\EE\Big[\lim_{n\to\infty} V^{2}\mathbf{1}_{|V|>1}(|V|^{-\frac{1}{n}}-1)^{2}\Big]=0.\]
Therefore, if we can show $V_n \in AC(\mathcal{G})$, then $V$ is in the closure of $AC(\cG)$ under the normal $\|\cdot\|_{\HH}$, namely $V\in \HH_0$, and the claim follows.
\par
We now show $V_n \in AC(\mathcal{G})$ for any $n\geq 1$. 	Notice
\[\EE|V_n|^{2+\frac{1}{n}} = \EE\big[|V|^{(1-\frac{1}{n})(2+\frac{1}{n})}\mathbf{1}_{|V|>1}+|V|^{2+\frac{1}{n}}\mathbf{1}_{|V|\leq 1}\big]\leq \EE (|V|^2\mathbf{1}_{|V|>1}+1) <\infty,\]
so $V_n \in L^{2+\frac{1}{n}}$.
By H\"{o}lder's inequality, we have
\[\EE|V_n \rho_T|^2 \le \(\EE |V_n|^{2\frac{2n+1}{2n}} \)^{\frac{2n}{2n+1}} \(\EE \rho_T^{2(2n+1)} \)^{\frac{1}{2n+1}}<\infty,\]
as $\EE \rho_T^p < \infty$, for all $ p>1$. Hence $\EE (V_n\rho_T| \mathcal{G}_t)$ is a square integrable martingale. By Theorem \ref{thm2.1}, we have 
\begin{eqnarray} 
\EE (V_n\rho_T| \mathcal{G}_t) - \EE (V_n\rho_0) = \int_0^t \Phi_s d\nu_s,\quad t\in[0,T],
\end{eqnarray}
for some $\Phi_s \in L^2_{\mathcal{G}}(0,T;\mathbb{R})$. 
When $t = T$, since $V_n \rho_T$ is $\mathcal{G}_T$ measurable, the above equation reduces to 
\begin{eqnarray}
V_n\rho_T- \EE (V_n\rho_0) = \int_0^T \Phi_s d\nu_s,
\end{eqnarray}
which implies $V_n \in AC(\mathcal{G})$. 
\end{proof}

 It was shown in \cite{XZ} that the optimal terminal $v$ to the optimization problem for the model in \cite{XZ} is given by 
\begin{eqnarray}\label{solution}
v = \frac{(z\langle \beta, \beta\rangle_{\mathbb{H}} - x_0\langle \alpha, \beta \rangle_{\mathbb{H}})\alpha + (-z \langle \alpha, \beta \rangle_{\mathbb{H}} + x_0 \langle \alpha, \alpha \rangle_{\mathbb{H}})\beta}{\langle \alpha, \alpha \rangle_{\mathbb{H}} \langle \beta, \beta \rangle_{\mathbb{H}}- \langle \alpha, \beta \rangle^2_{\mathbb{H}}},
\end{eqnarray}
where $\alpha, \beta$ are the orthogonal projections on $\mathbb{H}_0$ of $1$ and $\rho_T$, respectively.

A numerical scheme was given in \cite{XZ} under the completeness assumption. Although that numerical scheme can be extended to the current model, we will propose a more efficient numerical scheme for our model, which is one of the main contributions of the current article.

\subsection{Find the optimal strategy }
\numberwithin{equation}{section}

Similar to \cite{XZ}, the terminal problem can be solved as follows.
\begin{lemma}
The optimal terminal wealth for the problem \eqref{opt} is
\begin{eqnarray}\label{strategy}
v = \frac{z\EE(\rho_T^2)- y_0\EE\rho_T + (y_0-z\EE\rho_T )\rho_T}{\mathrm{Var} (\rho_T)},
\end{eqnarray}
where $\rho_T$ is given by \eqref{rho}.
\end{lemma}

After finding the optimal terminal wealth, we then seek a portfolio to realize it. Namely, for
 $v$ given by (\ref{strategy}), we need to find a solution of the following BSDE:
\begin{eqnarray} \label{BSDE11}
\begin{cases}
dY_t = \big(\(b + \ga \pi_t - r\)u_t + r Y_t \big)dt + u_t d\nu_t, \quad t\in[0,T],\\
Y_T = v.
\end{cases}
\end{eqnarray}

\begin{theorem}\label{thm:optimalcontrol}
The optimal portfolio is given by
\begin{eqnarray}\label{eq0903a}
u_t = (b-r +\ga \pi_t)Y_t + \rho^{-1}_t \eta_t,
\end{eqnarray}
where $\eta_t \in L_{\cG}^2(0,T;\RR^{d})$ is uniquely determined by the martingale representation theorem
\begin{eqnarray} \label{eta}
\EE(\theta|\cG_t) = \EE( \theta)+ \int_0^t \eta_sd \nu_s, \quad \quad \forall t\in [0,T],
\end{eqnarray} 
and $\theta = \rho_T Y_T$.
\end{theorem}
\begin{proof} 
As seen from the arguments above, we need to seek a solution to the following forward-backward SDE:
\begin{equation}
\begin{cases}\label{BSDE}
dY_t = \big((b +\ga \pi_t - r)u_t + r Y_t \big)dt + u_t d\nu_t,\quad Y_0=y_0,\\
d \pi_t = \ga \pi_t (1 - \pi_t) d\nu_t,\\
d\rho_t = -r \rho_t dt - (b -r +\ga\pi_t)\rho_t d\nu_t, \\
\rho_0 = 1, \quad \pi_0=c_0, \quad Y_T = c_1 + c_2\rho_T,
\end{cases}
\end{equation}
where $c_0=P(\mu=a)$, $c_1=\frac{z\EE \rho_T^2-y_0\EE\rho_T}{\mathrm{Var} (\rho_T)}$ and $c_2=\frac{y_0-z\EE\rho_T}{\mathrm{Var} (\rho_T)}$ are known constants.

To prove the invertibility of $\rho_t$, we define $\Phi_t$ by the following SDE:
\begin{equation}
\begin{cases}\label{SDE1}
d\Phi_t = (r + (b-r+\ga \pi_t)^2)\Phi_t dt + (b-r+\ga \pi_t)\Phi_t d\nu_t,\\
\Phi_0 = 1.
\end{cases}
\end{equation}
Apply It\^o's formula to $\rho_t\Phi_t$, we have
\[d(\rho_t\Phi_{t})=0,\]
and hence, $\rho_t\Phi_t\equiv\rho_0\Phi_0=1$. Since $\rho_t Y_t$ is a martingale, then 
\begin{eqnarray}\label{eq0903b}
Y_t = \rho^{-1}_t \EE(\rho_T Y_T | \cG_t)=\rho^{-1}_t \EE(\theta | \cG_t)=\Phi_t\EE(\theta | \cG_t).
\end{eqnarray}
Applying It\^o's formula to (\ref{eq0903b}) and comparing the result with (\ref{BSDE}), we obtain
 the expression (\ref{eq0903a}) for the optimal portfolio.
\end{proof}

Finding the numerical solution $(u_t,Y_t,\pi_t,\rho_t)$ of the FBSDE  (\ref{BSDE}) is the object of the next section.

\section{Numerical schemes based on Malliavin calculus}
\setcounter{equation}{0}

Pardoux and Peng  \cite{PP} obtained the unique solvability results for the nonlinear BSDEs in 1990. Since then a growing literature investigates the numerical methods for BSDEs (\cite{MPJ}, \cite{DJ}, \cite{BA}, \cite{MJ}, \cite{MP}, \cite{ZZ}, \cite{BG}, \cite{BG2}, \cite{GLW}, \cite{JZ}). In \cite{BET}, the Malliavin calculus approach and Monte Carlo approximation are employed to study the conditional expectation, in the Ph.D. thesis \cite{ZJD} of Zhang, some fine properties of the BSDEs by using Malliavin derivatives under weaker conditions were also studied. However, those works mentioned above are based on the setting that the drift coefficients of the BSDEs are deterministic. We can not apply these numerical schemes to our model directly.

In Xiong and Zhou's \cite{XZ} model, the coefficients of $u_t$ and $Y_t$ which appear in the drift term are random in general. They proposed a numerical approximation $(u^n_t,Y^n_t)$ to the solution $(u_t,Y_t)$ to that kind of BSDE with random coefficients. However, due to technical difficulty, only the convergence of $Y^n_t$ to the wealth process $Y_t$ is proved, and leave the convergence problem of the portfolio unsolved. The rate of convergence of $Y^n_t$ to $Y_t$ is not established in that paper.

In this section, we propose an efficient numerical scheme for the BSDE \eqref{BSDE11} whose terminal value $v$ takes the form $c_1 + c_2\rho_T$, where $c_1, c_2$ are constants and $\rho_t$ is a diffusion process which is Malliavin differentiable (see Theorem \ref{thm_malliavin} for detailed calculation). With the help of Malliavin calculus, we prove that our scheme for the portfolio and the wealth processes converge in the strong $L^2$ sense and derive the rate of convergence.

Denote $N(t):= \EE (\theta| \mathcal{G}_t)$. We note that the main complexity in Xiong and Zhou's \cite{XZ} numerical scheme for BSDEs results from the approximation of the integrand $\eta_t$ in \eqref{eta}, which is difficult to calculate directly. They use the following procedure to approximate $\eta_t$: 
First they divide $[0,t]$ into $n_1$ subintervals and approximate the quadratic covariation process
\begin{align*}
A_t:= \langle N, \nu \rangle_t = \int_0^t \eta_sds
\end{align*}
by the discrete version over the partition points. They further divide each subinterval mentioned above into $n_2$ smaller ones and obtain an approximation of $\eta_s,\;\ s\le t$. This procedure is not computationally efficient because the double-partition increases error dramatically. This will be seen from the numerical  examples in the subsequent section. 

In order to overcome the aforementioned drawback of the above numerical scheme, we turn to use the Clark-Ocone formula from Malliavin calculus to get an explicit expression of $\eta_t$. In fact, it will be the conditional expectation of a Malliavin derivative. Our numerical scheme will be based on this representation.

\begin{theorem}\label{thm_malliavin}
We can represent $\eta_t$ as $\EE \(D_t \theta | \cG_t \)$ where $D_t$ is the Malliavin derivative operator. Further,
\begin{equation}\label{eq0903c}
D_t\theta=(c_1+2c_2 \rho_T)D_t\rho_T 
\end{equation}
and $D_t\rho_T$ is given by
\begin{equation}\label{eq0903d}
D_t\rho_T=\rho_T\left[-\int_t^T\ga(b-r+\ga \pi_s)D_t\pi_s ds-(b-r+\ga \pi_t)+\int_t^T\ga D_t \pi_s d\nu_s\right],
\end{equation}
with
\begin{align}\label{eq0903e}
D_t\pi_s=\ga\pi_t(1-\pi_t)\exp\left(\int_t^s\ga (1-\pi_r)d\nu_r-\frac{1}{2}\int_t^s\ga^2 (1- 2\pi_r)^2dr\right).
\end{align}
\end{theorem}
\begin{proof}
Note that 
\[
\theta=\rho_T Y_T=c_1\rho_T+c_2\rho_T^2,
\]
and hence, (\ref{eq0903c}) follows by applying the Malliavin derivative on both sides.

As
\[
\rho_T=\exp\left(-\int_0^T[r+\frac{1}{2}(b-r+\ga \pi_s)^2]ds-\int_0^T(b-r+\ga \pi_s)d\nu_s\right),
\]
a direct calculation yields (\ref{eq0903d}).

Applying Malliavin derivative to both sides of the integral form of the identity (\ref{pi112}), we get
\begin{align}\label{linearSDE}
D_t\pi_s=\ga\pi_t(1-\pi_t)+\int_t^s\ga (1-2\pi_r)D_t\pi_r d\nu_r.
\end{align}
Then, (\ref{eq0903e}) follows by solving the linear SDE (\ref{linearSDE}). Finally, (\ref{eq0903c}) follows from the Clark-Ocone formula given in Section 2.
\end{proof}
\begin{remark}
Our method is based on the Malliavin  differentiability of $\rho_T$ so that the solution of \eqref{BSDE11} can be represented explicitly. In this paper, our setting is a drift uncertainty model with $\mu$ only takes two values, nevertheless,  the whole analysis can be generalized to a model with several risky assets, where the drift vector takes finite states, under which assumption we can still deduce the Malliavin differentiability of $\rho_T$.
\end{remark}

\subsection{A numerical scheme and its analysis}
Based on Theorem \ref{thm_malliavin}, it is easy to show that
\[\eta_t = \EE \(D_t \theta | \cG_t \):= N_1(t) +\ga N_2(t) + \ga N_3(t),\]
with $N_j(t)=\EE\(I_j|\cG_t\)$, $j=1,2,3$,
where
\begin{align}\label{I1}
I_1 &=-(c_1\rho_T+2c_2 \rho_T^2) (b-r+\ga \pi_t), \\
\label{I2}
I_2 &= (c_1\rho_T+2c_2 \rho_T^2)\int_t^T D_t \pi_s d\nu_s,\\
\label{I3}
I_3 &=-(c_1\rho_T+2c_2 \rho_T^2) \int_t^T (b-r+\ga \pi_s)D_t\pi_s ds,
\end{align}
and $D_t\pi_s$ is given by (\ref{eq0903e}).

As in the proof of Theorem \ref{thm:optimalcontrol}, the key to solve the optimal portfolio is the martingale representation of the $\cG_t$-martingale $\EE(\theta|\cG_t)$. We will establish particle representation for this martingale.

The solution of \eqref{rho11} is given by 
\begin{eqnarray}
\rho_t = \exp\(- \int_0^t (b-r+\ga \pi_s) d\nu_s -\int_0^t (r+ \frac{1}{2}(b-r+\ga \pi_s)^2)ds \),
\end{eqnarray}
Denote $\tilde{\rho}_t := \log \rho_t$, then
\begin{eqnarray}\label{logg}
d\tilde{\rho}_t =- (b-r+\ga \pi_t) d\nu_t -(r+ \frac{1}{2}(b-r+\ga \pi_t)^2)dt.
\end{eqnarray}

To approximate $\EE(\pi_t|\cG_{t'})$, we use the conditional SLLN such that $\pi^i_t$ is given by (\ref{pi112}) with $\nu_s$ be replaced by $\nu^i_s$ for $s\ge t'$, where $\nu^i$, $i=1,2,\cdots$ are independent copies of $\nu$. More precisely, we define the following processes $\pi^i(t,t')$ with two time-indices as follows: For $t\le t'$, $\pi^i(t,t')=\pi_t$, and for $t\ge t'$,
\begin{align}\
d \pi^i(t,t') =\ga \pi^i(t,t')(1 - \pi^i(t,t')) d\nu^i_t,\qquad \pi^i(t',t')=\pi(t').
\end{align}

To approximate $\EE(\rho_t|\cG_{t'})$, we use $\EE(\exp(\tilde{\rho}_t)|\cG_{t'})$ instead. For $t\le t'$, $\tilde{\rho}^i(t,t')=\tilde{\rho}_t$, and for $t\ge t'$,
\begin{align}\label{pirho12}
d\tilde{\rho}^i(t,t') =- (b-r+\ga \pi^i(t,t')) d\nu^i_t -(r+ \frac{1}{2}(b-r+\ga \pi^i(t,t'))^2)dt,\qquad \tilde{\rho}^i(t',t')=\tilde{\rho}(t').
\end{align}

By conditional SLLN, we can easily prove the following identities.
\begin{proposition} Denote $\rho^i (T,t)= \exp(\tilde{\rho}^i (T,t))$, we have
\begin{align*}
N_1(t) &= -( b-r + \ga \pi_t) 
\lim_{m \to \infty}\frac{1}{m} \sum_{i=1}^{m} ( c_1 \rho^i(T, t) + 2c_2 (\rho^i(T, t))^2), \\
N_2(t) &= \lim_{m \to \infty}\frac{1}{m} \sum_{i=1}^{m} ( c_1 \rho^i(T, t) + 2c_2 (\rho^i(T, t))^2) \int_t^T D_t \pi^i(s,t) d\nu^i_s, \\
N_3(t) &= \lim_{m \to \infty}\frac{1}{m} \sum_{i=1}^{m} - ( c_1 \rho^i(T, t) + 2c_2 (\rho^i(T, t))^2) \int_t^T (b-r+\ga \pi^i(s,t)) D_t \pi^i(s,t) ds. 
\end{align*}
\end{proposition}

In order to approximate $N_k(t),(k=1,2,3)$, we use the discrete Euler scheme to approximate $\pi_t$. For notation simplicity, from now on we assume $T = 1$. Then, we discrete the time interval $[0,1]$ into $n$ small intervals and let $\delta = \frac{1}{n}$.

Note that the diffusion coefficient in the SDE \eqref{pi112} is $\sigma(x) = \ga x(1-x)$, which does not satisfy the global Lipschitz condition \eqref{Hbeta2}. To overcome this hurdle, we define $\bar{\sigma}(x)$ as following
\begin{align*}
\bar{\sigma}(x) = 
\begin{cases}
\ga x(1 - x), \quad & x\in [0,1], \\
0, \quad & x \notin [0,1].
\end{cases}
\end{align*}
Using the fact that $\pi_{t}\in[0,1]$ for all $t\in[0,T]$, it is easy to see that $\pi_{t}$ is a solution of the following SDE 
\begin{align}\label{newpi}
d {\pi}_t =\bar{\sigma}( {\pi}_{t}) d\nu_t.
\end{align}
This SDE satisfies the global Lipschitz condition \eqref{Hbeta2}, so $\pi_{t}$ is the unique solution. Therefore, we approximate $\pi_t$ by applying the Euler scheme to \eqref{newpi} instead of the SDE \eqref{pi112}.

Firstly, we define $\pi^{i, \delta}(t, t'), t,t' \ge 0,$ in two steps.

For $l \le k$, let
\begin{align}
\pi^{\delta}(l\delta, k\delta) &:= \pi^{\delta}\((l-1)\delta, k\delta\)+ \bar{\sigma}(\pi^{\delta}\((l-1)\delta, k\delta\)) \(\nu_{l\delta} - \nu_{(l-1)\delta}\) \nonumber
\end{align}
with $\pi^{\delta}(0, k\delta) := c$ (c is a constant in $[0,1]$); for $l > k$, let
\begin{align}
\pi^{i,\delta}(l\delta, k\delta) &:= \pi^{i,\delta}\((l-1)\delta, k\delta\)+ \bar{\sigma}(\pi^{i, \delta}\((l-1)\delta, k\delta\))\(\nu^i_{l\delta} - \nu^i_{(l-1)\delta}\). \nonumber
\end{align}

For $l \le k$, let
\begin{align}
\rho^{\delta}(l\delta, k\delta) &:= \exp \left( \tilde{\rho}^{\delta}((l-1)\delta, k\delta) -\delta (r+ \frac{1}{2}(b-r+\ga \pi^{\delta}((l-1)\delta, k\delta))^2) \right. \nonumber \\
&\left. \quad\ - (b-r+\ga \pi^{\delta}((l-1)\delta, k\delta)) \(\nu_{l\delta} - \nu_{(l-1)\delta}\) \right );
\end{align}
for $l > k$, let
\begin{align}
\rho^{i,\delta}(l\delta, k\delta) &:= \exp \left( \tilde{\rho}^{i,\delta}((l-1)\delta, k\delta) -\delta (r+ \frac{1}{2}(b-r+\ga \pi^{i,\delta}((l-1)\delta, k\delta))^2) \right. \nonumber \\
&\left. \quad\ - (b-r+\ga \pi^{i,\delta}((l-1)\delta, k\delta)) \(\nu^i_{l\delta} - \nu^i_{(l-1)\delta}\) \right),
\end{align}
with $\tilde{\rho}^{\delta}_0 = 0$.

Similarly, denote $\tilde{\Phi}_t = \log \Phi_t$, and let 
\begin{align}
\tilde{\Phi}^{\delta}_{k\delta} &:= \tilde{\Phi}^{\delta}_{(k-1)\delta} + \delta \(r+ \frac{1}{2} \(b-r+\ga\pi^{\delta}((k-1)\delta, k\delta)\)^2\) \nonumber \\
&\quad\; + \(b-r+\ga\pi^{\delta}\((k-1)\delta, k\delta\)\) \(\nu_{k\delta} - \nu_{(k-1)\delta}\), \nonumber
\end{align}
with $\tilde{\Phi}^{\delta}_0 = 0$. Then $\Phi_{k\delta}^{\delta} = \exp (\tilde{\Phi}_{k\delta}^{\delta})$.

Next we approximate $N_i(t)$ by $N_i^{m, \delta}(k\delta)$, ($i = 1,2,3;$ $m$ is related to the SLLN, which will be chosen later). For all $s \in [t, T], t \in [0, T]$, let $k = [nt]$, $j = [ns]$. Then $t \in [k\delta,(k+1)\delta)$ and $s \in [j\delta,(j+1)\delta)$. We define $N_i^{m,\delta}(k\delta),(i =1,2,3)$ as follows:
\begin{align}
N_1^{m,\delta}(k\delta) &= -\( b-r + \ga \pi^{\delta}(k\delta) \) 
\frac{1}{m} \sum_{i=1}^{m} \( c_1 \rho^{i,\delta}(T, k\delta) + 2c_2 (\rho^{i,\delta}(T, k\delta))^2 \), \nonumber\\
N_2^{m,\delta}(k\delta) &= \frac{1}{m} \sum_{i=1}^{m} \( c_1 \rho^{i,\delta}(T, k\delta) + 2c_2 (\rho^{i,\delta}(T, k\delta))^2\) S_2^{i,\delta}(T, k\delta), \nonumber\\
N_3^{m,\delta}(k\delta) &= \frac{1}{m} \sum_{i=1}^{m} -\( c_1 \rho^{i,\delta}(T, k\delta) + 2c_2 (\rho^{i,\delta}(T, k\delta))^2\) S_3^{i,\delta}(T, k\delta) , \nonumber
\end{align}
where 
\begin{align*}
S_2^{i,\delta}(T, k\delta) &= \sum_{l = 1}^{n-k} D_{k \delta}\pi^{i,\delta}((l + k -1) \delta, k\delta) \(\nu^{i}_{l\delta} - \nu^{i}_{(l-1)\delta}\), \\
S_3^{i,\delta}(T, k\delta) &= \sum_{l = 1}^{n-k} \delta(b-r+\ga\pi^{i,\delta}((l + k -1) \delta, k\delta) ) D_{k \delta}\pi^{i,\delta}((l + k -1) \delta, k\delta).
\end{align*}

In the above, $D_{k\delta}\pi^{i,\delta}(j\delta,k\delta), (j = k,\cdots,n-1)$ are still stochastic integrals. By (\ref{linearSDE}), we define $D_{k\delta}\pi^{i,\delta}(j\delta,k\delta)$ only in one step.
Namely, for $j \ge k$, we define
\begin{align}
D_{k\delta}\pi^{i,\delta}(j\delta,k\delta) &:= D_{k\delta}\pi^{i,\delta}((j-1)\delta,k\delta) \nonumber \\
&\quad\; + \ga \(1 - 2 \pi^{i,\delta}((j-1)\delta,k\delta) \)D_{k\delta}\pi^{i,\delta}((j-1)\delta,k\delta)\(\nu^i_{j\delta} - \nu^i_{(j-1)\delta}\) \nonumber
\end{align}
with $D_{k\delta}\pi^{i,\delta}(k\delta,k\delta) =\ga \pi_{k\delta}(1 -\pi_{k\delta})$.

Finally, we obtain
\begin{align}\label{approx}
\eta^{\delta, m}_{k\delta} &= N_1^{m, \delta}(k\delta) + \ga N_2^{m, \delta}(k\delta) + \ga N_3^{m, \delta}(k\delta) .
\end{align} 

To summarize, we can approximate $Y_t$ and $ u_t$, $k\de\le t<(k+1)\de$, by $Y_{k\delta}^{\delta,m}$ and $u_{k\delta}^{\delta,m}$, where
\begin{align}\label{olY}
Y_{k\delta}^{\delta,m} = \Phi_{k\delta}^{\delta} \frac{1}{m} \sum_{i=1}^{m}\(c_1\rho^{i,\delta}(T,k\delta)+ c_2 \(\rho^{i,\delta}(T,k\delta)\)^2 \)
\end{align}
and
\begin{align}\label{Appu}
u_{k\delta}^{\delta,m} = (b - r + \ga \pi^{\delta}_{k\delta})Y_{k\delta}^{\delta,m} + \Phi_{k\delta}^{\delta} \eta^{\delta, m}_{k\delta}.
\end{align}

\begin{theorem}\label{error analysis}
There exists a constant $C$ such that for any $k\de\le T=1$, we have
\begin{align*}
\|u_{k\delta} - u^{\delta,m}_{k\delta}\|_2 \le C\(\sqrt{\delta} + \frac{1}{\sqrt{m}}\) 
\end{align*}
and 
\begin{align*}
\|Y_{k\delta} - Y^{\delta,m}_{k\delta}\|_2 \le C\(\sqrt{\delta} + \frac{1}{\sqrt{m}}\).
\end{align*}
\end{theorem}

\begin{proof}
Since we apply the Euler scheme for the new equation \eqref{newpi} which satisfies all the conditions in Theorem \ref{Euler}. Thus, 
\begin{align*}
\|\pi - \pi^{\delta}\|_4 \le C\sqrt{\delta}.
\end{align*}
Besides, since $\Phi_t$, $\rho_t$ and $D_t\rho$ are given by exponential stochastic integrals, then by the Burkholder-Davis-Gundy inequality and H\"{o}lder's inequality, we have
\begin{align*}
\|\Phi - \Phi^{\delta}\|_4 \le C\sqrt{\delta}, \quad \| v\rho - v^{\delta}\rho^{\delta}\|_4 \le C\sqrt{\delta}, \quad 	\|D_t\rho - D_t\rho^{\delta}\|_4 \le C\sqrt{\delta}.
\end{align*}
From the representation (\ref{eq0903a}) and the approximation (\ref{Appu}), we first estimate the error between $u_{k\delta}$ and $u^{\delta,m}_{k\delta}$, 
\begin{align} 
\Big\|u_{k\delta} - u^{\delta,m}_{k\delta}\Big\|_2 
&\le\bigg\|\Phi_{k\delta}(b-r+\ga\pi_{k\delta})\EE\(v(T,k\delta)\rho(T,k\delta)|\cG_{k\delta}\)\nonumber\\
&\quad\quad\quad - \Phi^{\delta}_{k\delta}(b-r+\ga\pi^{\delta}_{k\delta}) \frac{1}{m}\sum_{i=1}^{m}(v^{i,\delta}(T,k\delta)\rho^{i,\delta}(T,k\delta)\bigg\|_2 \nonumber\\
&\quad\;+\bigg\|\Phi_{k\delta}\EE(c_1D_{k\delta}\rho(T,k\delta)+2c_2\rho(T,k\delta)D_{k\delta}\rho(T,k\delta)|\cG_{k\delta}) \nonumber\\
&\quad\quad\quad -\Phi^{\delta}_{k\delta}\frac{1}{m}\sum_{i=1}^{m}\(c_1D_{k\delta}\rho^{i,\delta}(T,k\delta)+2c_2\rho^{i,\delta}(T,k\delta)D_{k\delta}\rho^{i,\delta}(T,k\delta)\)\bigg\|_2 \nonumber\\
&:= J_1 + J_2, 
\end{align}
where 
\begin{align}\label{J1}
J_1 &= \Big\|\Phi_{k\delta}(b-r+\ga\pi_{k\delta})\EE\(v(T,k\delta)\rho(T,k\delta)|\cG_{k\delta}\)\nonumber \\
&\quad\quad- \Phi^{\delta}_{k\delta}(b-r+\ga\pi^{\delta}_{k\delta}) \frac{1}{m}\sum_{i=1}^{m}(v^{i,\delta}(T,k\delta)\rho^{i,\delta}(T,k\delta)\Big\|_2 \nonumber \\
&\le \Big\|\Phi_{k\delta}(b-r+\ga\pi_{k\delta})\EE\(v(T,k\delta)\rho(T,k\delta)|\cG_{k\delta}\)\nonumber\\
&\quad\quad\quad - \Phi^{\delta}_{k\delta}(b-r+\ga\pi^{\delta}_{k\delta})\EE\(v(T,k\delta)\rho(T,k\delta)|\cG_{k\delta}\)\Big\|_2 \nonumber \\
& \quad\;+ \Big\| \Phi^{\delta}_{k\delta}(b-r+\ga\pi^{\delta}_{k\delta})\EE\(v(T,k\delta)\rho(T,k\delta)|\cG_{k\delta}\) \nonumber\\
&\quad\quad\quad - \Phi^{\delta}_{k\delta}(b-r+\ga\pi^{\delta}_{k\delta}) \frac{1}{m}\sum_{i=1}^{m}(v^{i,\delta}(T,k\delta)\rho^{i,\delta}(T,k\delta)\Big\|_2 \nonumber \\
&\le \Big\| \Phi_{k\delta}(b-r+\ga\pi_{k\delta}) - \Phi^{\delta}_{k\delta}(b-r+\ga\pi^{\delta}_{k\delta})\Big\|_4 \times\Big\|\EE\(v(T,k\delta)\rho(T,k\delta)|\cG_{k\delta}\) \Big\|_4 \nonumber \\
&\quad\;+ \Big\|\Phi^{\delta}_{k\delta}(b-r+\ga\pi^{\delta}_{k\delta})\Big \|_4 \times\Big\| \EE\(v(T,k\delta)\rho(T,k\delta)|\cG_{k\delta}\) )- \EE\(v^{\delta}(T,k\delta)\rho^{\delta}(T,k\delta)|\cG_{k\delta}\) \Big\|_4 \nonumber \\
&\quad\;+\Big \|\Phi^{\delta}_{k\delta}(b-r+\ga\pi^{\delta}_{k\delta}) \Big\|_4\times \Big\| \EE\(v^{\delta}(T,k\delta)\rho^{\delta}(T,k\delta)|\cG_{k\delta}\) - \frac{1}{m}\sum_{i=1}^{m}(v^{i,\delta}(T,k\delta)\rho^{i,\delta}(T,k\delta) \Big\|_4 \nonumber \\
&\le C\(\sqrt{\delta} + \frac{1}{\sqrt{m}}\), 
\end{align}
and
\begin{align}\label{J2}
J_2 &\le \Big\| \Phi_{k\delta} - \Phi^{\delta}_{k\delta}\Big\|_4 \times\Big\|\EE(c_1D_{k\delta}\rho(T,k\delta)+2c_2\rho(T,k\delta)D_{k\delta}\rho(T,k\delta)|\cG_{k\delta}) \Big\|_4 \nonumber \\
&\quad\;+ \Big\|\Phi^{\delta}_{k\delta} \Big\|_4 \times\Big\| \EE(c_1D_{k\delta}\rho(T,k\delta)+2c_2\rho(T,k\delta)D_{k\delta}\rho(T,k\delta)|\cG_{k\delta}) \nonumber\\
&\quad\quad\quad\quad\quad\quad\quad\quad - \EE(c_1D_{k\delta}\rho^{\delta}(T,k\delta)+2c_2\rho^{\delta}(T,k\delta)D_{k\delta}\rho^{\delta}(T,k\delta)|\cG_{k\delta}) \Big\|_4 \nonumber \\
&\quad\;+ \Big\|\Phi^{\delta}_{k\delta} \Big\|_4 \times\Big \| \EE(c_1D_{k\delta}\rho^{\delta}(T,k\delta)+2c_2\rho^{\delta}(T,k\delta)D_{k\delta}\rho^{\delta}(T,k\delta)|\cG_{k\delta}) \nonumber\\
&\quad\quad\quad\quad\quad\quad\quad\quad - \frac{1}{m}\sum_{i=1}^{m}\(c_1D_{k\delta}\rho^{i,\delta}(T,k\delta)+2c_2\rho^{i,\delta}(T,k\delta)D_{k\delta}\rho^{i,\delta}(T,k\delta)\) \Big \|_4 \nonumber \\
&\le C\(\sqrt{\delta} + \frac{1}{\sqrt{m}}\).
\end{align}
By (\ref{J1}) and (\ref{J2}), we have 
\begin{align} 
\Big\|u_{k\delta} - u^{\delta,m}_{k\delta}\Big\|_2 \le C\(\sqrt{\delta} + \frac{1}{\sqrt{m}}\).\nonumber
\end{align}
Similarly, we can prove 
\begin{align} 
\Big\|Y_{k\delta} - Y^{\delta,m}_{k\delta}\Big\|_2 \le C\(\sqrt{\delta} + \frac{1}{\sqrt{m}}\),
\end{align}
which converges to 0 if we take $m=n$ (in this case, $\delta=\frac{1}{m}$).
\end{proof}
\begin{remark}
In Xiong and Zhou's \cite{XZ} model, only the convergence of $Y^n_t$ to $Y_t$ is proved, however, due to technical difficulty, the rate of convergence of $Y^n_t$ to $Y_t$ is not established in that paper. The above theorem not only proves the convergence of $Y^n_t$ to $Y_t$, but also gives the rate of convergence. Note that the errors in our numerical scheme consist of the error from the Euler approximation and that from the SLLN only. From this point of view, under the drift uncertainty model, the numerical scheme we proposed is more efficient than that of \cite{XZ}. 
\end{remark}

\subsection{Numerical results}
In  last subsection, we proved the convergence of the numerical scheme. Now we use Matlab to give an example to compare our method with that of \cite{XZ}. To this end, we first construct a BSDE whose drift term is random, and the solution $(X(t), Z(t))$ can be computed explicitly. Then we obtain the numerical approximations for $(X(t), Z(t))$ by applying two different  schemes. Next, we compare these two approximations with the actual solution. Finally, we apply the new algorithm to simulate the drift uncertainty model. 

For convenience, we denote the numerical method proposed by Xiong and Zhou \cite{XZ} as the ``old algorithm'', the one proposed by us as the  ``new algorithm'' and the explicit solution as the ``true value''. 

For the appearance of presentation, we now write $W_{t}$ as $W(t)$ (especially when $t$ has a complicated expression).
Let 
\[
H(t) = \int_0^t (1+W(s))dW(s) - \int_0^t (W(s)^2 + 2W(s))ds.
\]
We consider a BSDE with random coefficients as following 
\begin{eqnarray} \label{example}
\begin{cases}
dX(t) = \Big(-\frac{1}{2}\big(1-2W(t)-W(t)^2\big)X(t) - (1+W(t))Z(t)\Big)dt + Z(t)dW(t),\\
X(T) = \exp\big(H(T) - 2T\big), \quad\quad t\in[0,T].
\end{cases}
\end{eqnarray}
Following the steps (see Theorem 2.2, Chapter 7 of \cite{YZ}), we can solve the above linear BSDE \eqref{example} explicitly as 
\begin{align}\label{real}
X(t) = e^{H(t) -2t}, \quad Z(t) = (1+W(t))X(t).
\end{align}

For simplicity, let $T = 1$. We discrete $[0,1]$ into $n_1n_2$ small intervals. Denote $\delta_1 = \frac{1}{n_1}$ and $\delta_2 = \frac{1}{n_1n_2}$.

Let $\theta = \exp\Big(2\int_0^T (1+W(s))dW(s) - 2\int_0^T (1 + W(s))^2ds\Big)$ and $\Phi(t) = e^{-H(t)}$. Using the old algorithm, the approximation for $(X(t), Z(t))$ is given by
\begin{align*}
\begin{cases}
X^{m,\delta_1}(t) = \Phi^{\delta_1}(t)N^{m,\delta_1}(t),\\
Z^{m,\delta_1}(t) = -X^{m,\delta_1}(t) + \Phi^{\delta_1}(t)\eta_1^{m,\delta_1}(t),
\end{cases}
\end{align*}
where
\[
N(t) = \EE(\theta | \mathcal{F}_t^{W}),
\]
and $\Phi^{\delta_1}(t)$ is the approximation of $\Phi(t)$ generated by the Euler scheme, $N^{m,\delta_1}(t)$ is the approximation of $N(t)$ generated by the particle representation as well as the Euler scheme, and 
\begin{align}\label{Phi_ndelta}
\eta_1^{m,\delta_1}\(\frac{k}{n_1}\) &:= n_1 \sum_{j=1}^{n_2}\( N^{m,\delta_2}\(\frac{k-1}{n_1} + \frac{j}{n_1n_2}\) - N^{m,\delta_2}\(\frac{k-1}{n_1} + \frac{j-1}{n_1n_2}\) \) \nonumber \\
&\quad\;\times \( W^{\delta_2}\(\frac{k-1}{n_1} + \frac{j}{n_1n_2}\) - W^{\delta_2}\(\frac{k-1}{n_1} + \frac{j-1}{n_1n_2}\) \) , \\
& k = 1,2,\cdots, n_1 -1, \quad n_1, n_2 = 1,2,\cdots.\nonumber
\end{align}

On the other hand, since $\theta$ is Malliavin differentiable, we get 
\begin{align}\label{meta}
\eta_2(t) &:= \EE\(D_t\theta | \mathcal{F}_t^{W}\) \nonumber\\
& = \EE\(e^{2\int_t^T (1+W(s))dW(s) - 2\int_t^T (1 + W(s))^2ds}\Big[2W(T) - 4\int_t^T(1+W(s))ds + 2\Big] \Big| \mathcal{F}_t^{W}\) .
\end{align}
By the new algorithm, the approximation for $(X(t), Z(t))$ is given by
\begin{align*}
\begin{cases}
X^{m,\delta_2}(t) = \Phi^{\delta_2}(t)N^{m,\delta_2}(t),\\
Z^{m,\delta_2}(t) = -X^{m,\delta_2}(t) + \Phi^{\delta_2}(t)\eta_2^{m,\delta_2}(t),
\end{cases}
\end{align*}
where $\eta_2^{m,\delta_2}(t)$ is the approximation of $\eta_2(t)$ generated by the particle representation as well as the Euler scheme.

Using the aforementioned algorithms, we generate the following figures.

\bigskip
\begin{figure}[H]
\centering
\begin{tabular}{cc}
\begin{minipage}{0.45\linewidth}
\includegraphics[width=\textwidth]{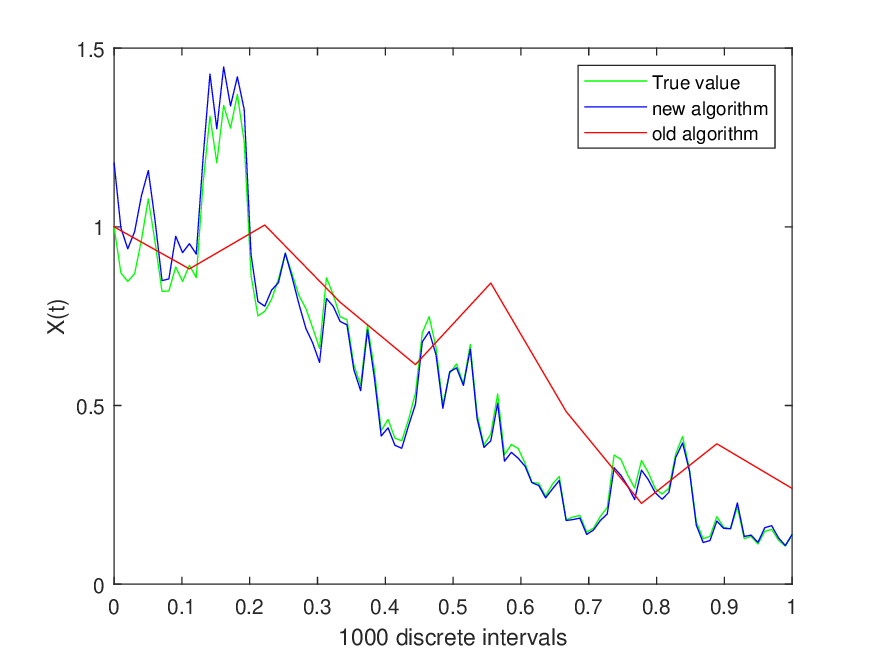}
\caption{$X(t)$ with $100$ discrete intervals}
\label{figure1}
\end{minipage}
\qquad
\begin{minipage}{0.45\linewidth}
\includegraphics[width=\textwidth]{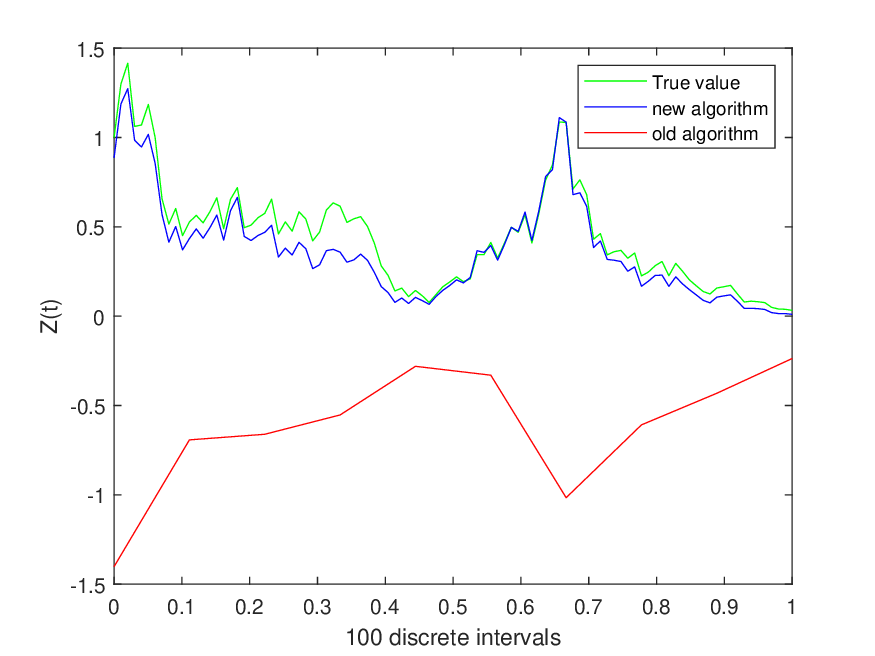}
\caption{$Z(t)$ with $100$ discrete intervals}
\label{figure2}
\end{minipage}
\end{tabular}
\end{figure}

\begin{figure}[H]
\centering
\begin{tabular}{cc}
\begin{minipage}{0.45\linewidth}
\includegraphics[width=\textwidth]{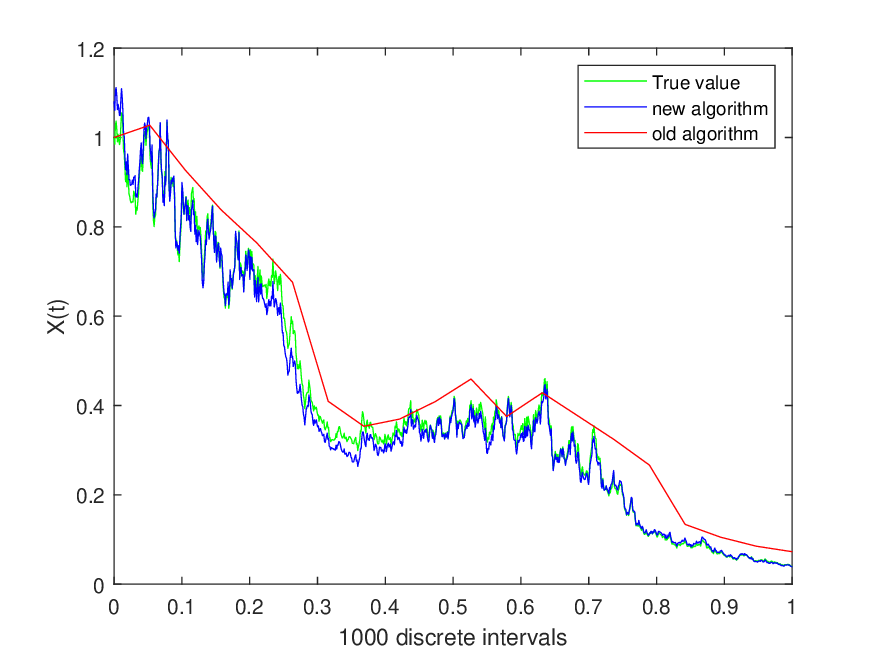}
\caption{$X(t)$ with $10^3$ discrete intervals}
\label{figure3}
\end{minipage}
\qquad
\begin{minipage}{0.45\linewidth}
\includegraphics[width=\textwidth]{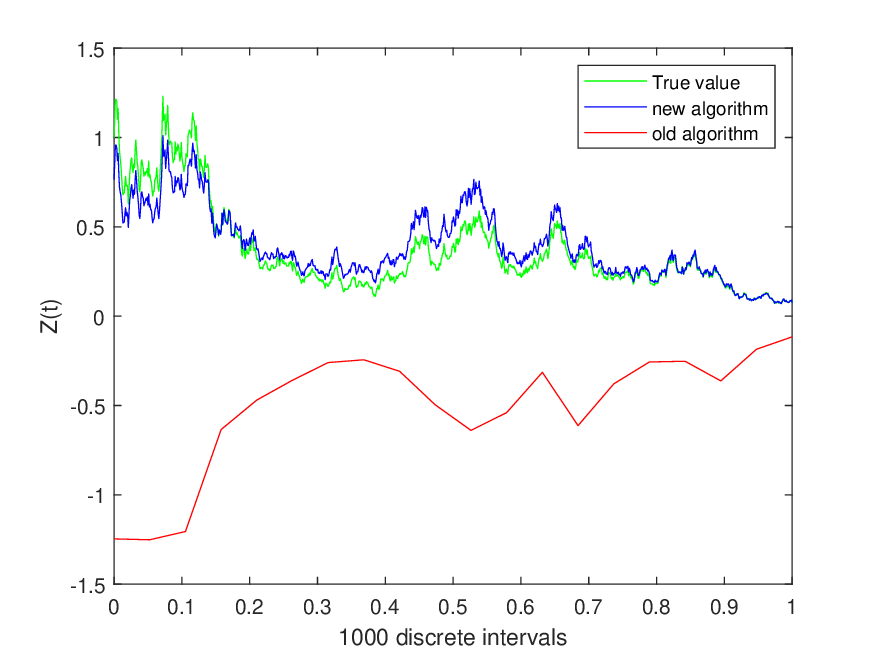}
\caption{$Z(t)$ with $10^3$ discrete intervals}
\label{figure4}
\end{minipage}
\end{tabular}
\end{figure}

It can be seen from Figures \ref{figure1}, \ref{figure2}, \ref{figure3}, \ref{figure4} that our new numerical scheme well approximated the true processes $X(t)$ and $Z(t)$. The curves of $X(t)$ and $Z(t)$ generated by the new  scheme are almost the same as the true processes. In contrast, the paths generated by the old  scheme are relatively far off due to the double-partition which is the main drawback in that scheme. The new algorithm overcomes this drawback by using the Malliavin calculus.

Finally, we apply our efficient numerical scheme to simulate the wealth process $Y_t$ and the optimal control $u_t$ for the drift uncertainty model.

We choose the parameters as following: $n =1000$, $\delta = \frac{1}{1000}$, $m = 1000$, $r = 0.03$, $a= 0.04$, $b=0.032$, $y_0 = 100$, $z = y_0\cdot(1+r+0.03)$, and $\pi_0 = 0.1$. 

Figure \ref{figure5} below is the numerical results for the innovation process $\nu_t$, the wealth process $Y_t$ and the optimal control  $u_t$. 

\begin{figure}[H]
\centering
\includegraphics[width=1\textwidth]{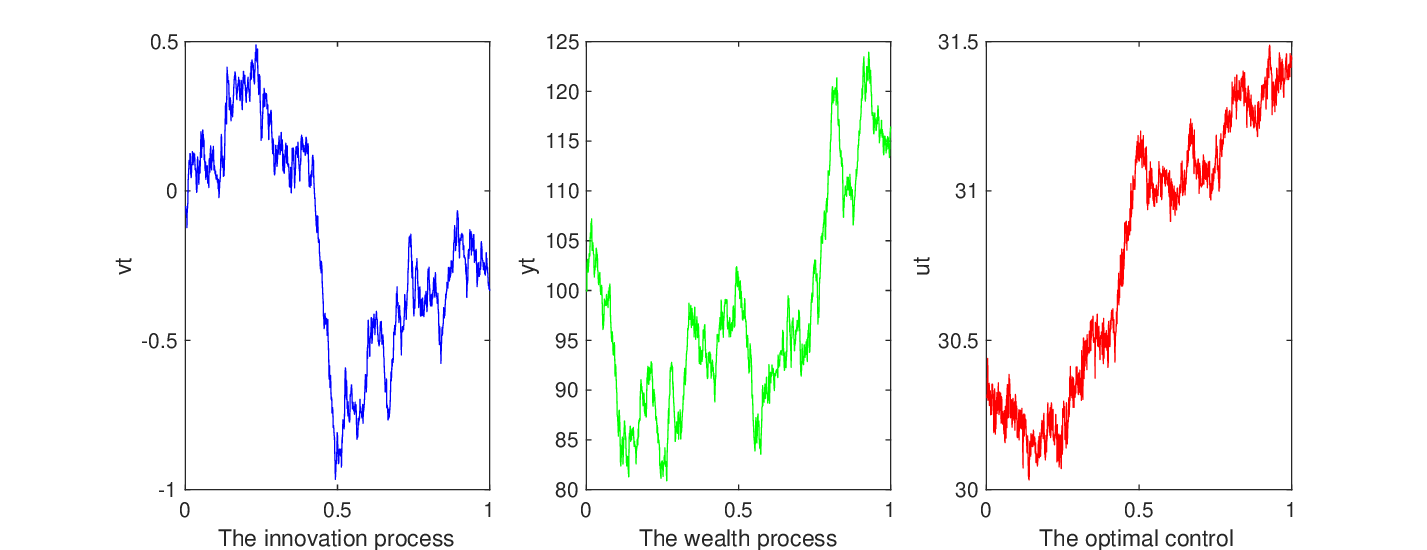}
\caption{drift uncertainty model with $10^3$ discrete intervals. }
\label{figure5}
\end{figure}

\end{document}